\def\year{2018}\relax
\documentclass[letterpaper]{article} \usepackage{aaai18}  \usepackage{times}  \usepackage{helvet}  \usepackage{courier}  \usepackage{url}  \usepackage{graphicx}  \usepackage{booktabs}
\usepackage{wrapfig}

\usepackage{numprint}
\usepackage{paralist}
\usepackage{color}
\usepackage{standalone}
\usepackage{tikz}

\newcommand{\citet}[1]{\citeauthor{#1}~\shortcite{#1}}
\newcommand{\citep}{\cite}

\usepackage{mathtools}
\usepackage{amsmath}
\usepackage{amsthm}
\usepackage{amssymb}
\usepackage{enumitem}

\newcommand{\cI}{\ensuremath{\mathcal{I}}}

\newcommand{\be}{\begin{eqnarray}}
\newcommand{\ee}[1]{\label{#1}\end{eqnarray}}

\newcommand{\ese}{\end{eqnarray*}}
\newcommand{\bse}{\begin{eqnarray*}}
\def\beq{\begin{equation}}
\def\eeq{\end{equation}}

\def\fnote#1{\footnote}

\def\*{{{\LARGE\bf $^*$}}}

 \def\seq{{\mathop{\sigma}}}
\def\infoset{{\mathop{\rm inf}}}

\newtheorem{theorem}{Theorem}

\newtheorem{definition}{Definition}

\title{Robust Stackelberg Equilibria in Extensive-Form Games\\ and Extension to Limited Lookahead}

\author{Christian Kroer \and Gabriele Farina \and Tuomas Sandholm\\
  Computer Science Department, Carnegie Mellon University\\
\{ckroer,gfarina,sandholm\}@cs.cmu.edu}

\begin{document}

\maketitle

\begin{abstract}
  Stackelberg equilibria have become increasingly important as a solution concept
in computational game theory, largely inspired by practical problems such as
security settings. In practice, however, there is typically uncertainty
regarding the model about the opponent. This paper is, to our knowledge, the
first to investigate Stackelberg equilibria under uncertainty in extensive-form
games, one of the broadest classes of game. We introduce robust Stackelberg
equilibria, where the uncertainty is about the opponent's payoffs, as well as
ones where the opponent has limited lookahead and the uncertainty is about the
opponent's node evaluation function.
We develop a new mixed-integer program for the deterministic limited-lookahead
setting. We then extend the program to the robust setting for Stackelberg
equilibrium under unlimited and under limited lookahead by the opponent. We show
that for the specific case of interval uncertainty about the opponent's payoffs
(or about the opponent's node evaluations in the case of limited lookahead),
robust Stackelberg equilibria can be computed with a mixed-integer program that
is of the same asymptotic size as that for the deterministic setting.
 \end{abstract}

In a Stackelberg equilibrium, a \emph{leader} commits to a strategy first, and then a \emph{follower} chooses a strategy for herself. The model was first introduced in the context of competition between firms where the leader picks a quantity to supply, and then the follower picks a quantity to supply~\cite{Stackelberg34:Marktform}. Stackelberg equilibria have become important as a solution concept
in computational game theory, largely inspired by practical problems such as
security settings, where the leader is a defender who picks a mixed (i.e., potentially randomized) strategy first, and then the follower who is the attacker decides where to attack, if at all. 

Most work on Stackelberg equilibria has focused on \emph{normal-form} (aka. matrix-form) games. \citet{Conitzer06:Computinga} studied the
problem of computing an optimal strategy to commit to in normal-form games. That line of work has been extended to many security-game applications. 
In practice, there is typically uncertainty about the opponent's payoffs. 
In normal-form games this has been studied as \emph{Bayesian Stackelberg
games} where the players have private information about their own payoffs, and there is common knowledge of the prior distribution over the payoffs~\citep{Tambe11:Security,Paruchuri08:Playing}.
As an alternative, the \emph{robust} (distribution-free) approach has been suggested for security games: bounds are assumed on the follower's payoffs~\citep{Kiekintveld13:Security,Nguyen14:Regret}.

\emph{Extensive-form games (EFGs)}---i.e., tree-form games---are a very general game representation language. EFGs are exponentially more compact and also more expressive than normal-form games. \citet{Letchford10:Computing} study how to compute an optimal strategy to commit
to in EFGs and prove hardness results under several assumptions about the game
structure. \citet{Bosansky15:Computation} provide further results specifically
for perfect-information EFGs. \citet{Bosansky15:Sequence} develop a \emph{mixed-integer program (MIP)} for
computing a Stackelberg strategy, and \citet{Cermak16:Using} develop an iterative
approach based on upper-bounding solutions from extensive-form \emph{correlated}
Stackelberg equilibria. 

To our knowledge, we are the first to
consider uncertainty about the opponent in Stackelberg strategies for EFGs. This is important because EFGs are a powerful representation language and because in practice there is typically uncertainty about the opponent. We take a robust approach to modeling this uncertainty. We
introduce robust Stackelberg equilibria for EFGs, where the uncertainty is about the opponent's payoffs, as well as ones where the opponent has limited lookahead and the uncertainty is about the opponent's node evaluation function.

We develop a new MIP for the deterministic limited-lookahead
setting. We then extend the MIP to the robust setting for
Stackelberg equilibrium under unlimited and under limited lookahead by the opponent. We show that for the
specific case of interval uncertainty about the opponent's payoffs (or about the opponent's node
evaluations in the case of limited lookahead), robust Stackelberg equilibria can be
computed with a MIP that is of the same asymptotic size as that for
the deterministic setting.

Our results for robust Stackelberg equilibria in EFGs are relevant to security-game settings with sequential interactions, where EFG models can more compactly represent certain games, as compared to a normal-form representation~\citep{Bosansky15:Sequence}.
Robust models are important in security games, where opponent models often have uncertainty, both in standard security games~\citep{Kiekintveld10:Robust,Kiekintveld13:Security,Nguyen14:Regret}, and \emph{green security games}~\citep{Nguyen15:Making}.

Our limited-lookahead results are useful for settings where it
is not always desirable to model adversaries as fully rational, but as having limited lookahead capability. This includes
settings such as biological games, where the goal is to steer an evolutionary process or
an adaptation process which typically acts myopically without lookahead~\citep{Sandholm15:Medical,Kroer16:Sequential} and security games where
opponents are often assumed to be myopic (which can be especially well motivated when the number of
adversaries is large~\citep{Yin12:TRUSTS} or in the case of \emph{opportunistic criminals}~\citep{Zhang16:Keeping,Rosenfeld17:Traffic}).
Our model of limited lookahead is an extension of that of
\citet{Kroer15:Limited} to a robust setting. \citet{Kroer15:Limited} gave a MIP
for computing an optimal strategy to commit to in the deterministic setting. We
show an alternative MIP for computing such a strategy to commit to, which we
then extend to the robust setting.

Finally, the question of robust variants of optimization problems has been
studied extensively in the optimization
literature~\citep{Tal02:Robust,Tal09:Robust,Bertsimas11:Theory}. In that
literature, the assumption is that we are given some \emph{nominal} mathematical
program, and then the robust variant requires that each constraint in the
nominal program holds with respect to every instantiation of a set of
uncertainty parameters. This makes the setting substantially different from our
setting, where there is no nominal program: the best response of the follower
does not need to be a best response for every uncertainty instantiation (this
would be the equivalent to robust optimization, and often infeasible), but
rather the best response is chosen after the uncertainty parameters are chosen.
 
\section{Extensive-Form Games}
Extensive-form games (EFGs) can be thought of as a game tree, where each node in
the tree corresponds to some history of actions taken by all players. Each node
belongs to some player, and the actions available to the player at a given node
are represented by the branches. Uncertainty is modeled by having a special
player, \emph{Nature}, that moves with some predefined fixed probability
distribution over actions at each node belonging to Nature. EFGs model imperfect
information by having groups of nodes in \emph{information sets}, where an information set  is a group of
nodes all belonging to the same player such that the player cannot distinguish
among them. Finally we assume \emph{perfect recall}, which requires that no
player forgets information they knew earlier in the game.

\begin{definition}
	A \emph{ leader-follower two-player extensive-form game with imperfect information and perfect recall} $\Gamma$ is a tuple $(H, Z, A, P, f_c, \mathcal{I}_l, \mathcal{I}_f, u_l,u_f)$ composed of:
	\begin{itemize}[nolistsep,itemsep=1mm]
	\item $H$: a finite set of possible sequences (or histories) of actions, such that the empty sequence $\varnothing \in H$, and every prefix $z$ of $h$ in $H$ is also in $H$.
	\item $Z\subseteq H$: the set of terminal histories, i.e., those sequences that are not a proper prefix of any sequence.
	\item $A$: a function mapping $h \in H\setminus Z$ to the set of available actions at non-terminal history $h$.
	\item $P$: the player function, mapping each non-terminal history $h \in H\setminus Z$ to $\{l, f, c\}$, representing the player whose turn it is to move after history $h$. If $P(h)=c$, the player is Chance.
	\item $\mathcal{C}$: a function assigning to each $h\in H$ the probability of
    reaching $h$ due to nature (i.e. assuming that both players play to reach
    $h$).
	\item $\mathcal{I}_i$, for $i\in\{l,f\}$: partition of $\{h\in H: P(h)=i\}$ with the property that $A(h)=A(h')$ for each $h,h'$ in the same set of the partition. For notational convenience, we will write $A(I)$ to mean $A(h)$ for any of the $h\in I$, where $I\in \mathcal{I}_i$. ${\cal I}_i$ is the information partition of player $i$, while the sets in ${\cal I}_i$ are called the information sets of player $i$.
	\item $u_i$: utility function mapping $z\in Z$ to the utility 
    gained by player $i$ when the terminal history is reached.
	\end{itemize}
	We further assume that all players have perfect recall.
\end{definition}
We will use the more relaxed term \emph{extensive-form game}, or EFG, to mean a two-player extensive-form game with imperfect information and perfect recall.

In this paper we will investigate settings where there is uncertainty
about the follower's utility function $u_f$. Specifically, the follower's utility
can be any function from some given \emph{uncertainty set} $U_f$ consisting of
functions that map from the set of leaf nodes to $\mathbb{R}$. We leave
the exact structure of $U_f$ undefined for now; in our algorithmic section we
show that the case where each leaf has independent interval uncertainty can be
solved using a MIP.

A strategy for a player $i$ is usually represented in \emph{behavioral form},
which consists of probability distributions over actions at each information set
in $\mathcal{I}_i$. In this paper we will focus on an alternative, but
strategically equivalent, representation of the set of strategies, called the
\emph{sequence
  form}~\citep{Romanovskii62:Reduction,Koller96:Efficient,Stengel96:Efficient}.
In the sequence form, actions are instead represented by \emph{sequences}. A
sequence $\sigma_i$, is an ordered list of actions taken by player $i$ on the
path to some history $h$. In perfect-recall games, all nodes in an information
set $I\in \mathcal{I}_i$ correspond to the same sequence for player $i$. We let
$\seq(I)$ denote this sequence. Given a sequence $\sigma_i$ and an action $a$
that Player $i$ can take immediately after $\sigma_i$, we let $\sigma_ia$ denote
the resulting new sequence. The set of all sequences for player $i$ is denoted
by $\Sigma_i$. Instead of directly choosing the probability to put on an action,
in the sequence form the probability of playing the entire sequence is chosen; 
this is called the \emph{realization probability} and is denoted by
$r_i(\sigma_i)$. A choice of realization probabilities for every sequence
belonging to Player $i$ is called a \emph{realization plan} and is denoted by 
$r_i:\Sigma_i \rightarrow [0,1]^{|\Sigma_i|}$. This representation relies on
perfect recall: for any information set $I\in \mathcal{I}_i$ we have that each
action $a\in A(I)$ is uniquely represented by a single sequence $\sigma_i =
\seq(I)a$, since $\seq(I)$ corresponds to exactly one sequence. This gives us a simple way to convert any strategy in sequence form to a
behavioral strategy: the probability of playing action $a\in A(I)$ at
information set $I$ is simply $\frac{r_i(\seq(I)a)}{r_i(\seq(I))}$. For a
sequence $\sigma=\sigma'a$, we let the information set such that $a \in A(I),
\seq(I)=\sigma'$ be denoted by $\infoset(\sigma)$.

It will be convenient to have function expressing expected values for a given
pair of sequences. Given two sequences $\sigma_l$ and $\sigma_f$, we let
\[
  g_l(\sigma_l,\sigma_f)=\sum_{h \in
  Z;\seq_f(h)=\sigma_f;\seq_l(h)=\sigma_l}\mathcal{C}(h)u_l(h),\]\[
  g_f^{u_f}(\sigma_l,\sigma_f)=\sum_{h \in
  Z;\seq_f(h)=\sigma_f;\seq_l(h)=\sigma_l}\mathcal{C}(h)u_f(h)
\]
be the expected utilities, for the leader and follower respectively, over leaf nodes that are reached with $\sigma_f,$ and
$\sigma_l$ as the corresponding last player sequences. The function
for the follower $g_f^{u_f}$ depends on the choice of utility function $u_f$,
whereas we always know the utility function for the leader.\footnote{In Stackelberg equilibrium, the follower does not have to be concerned about the leader's utility function because the leader commits to his strategy and declares his strategy to the follower.} Given two
realization plans $r_l,r_f$ and a utility function $u_i$, we overload notation
slightly and let the expected value for Player $i$ induced by the realization
plans be denoted by
\[
u_i(r_l,r_f) = \sum_{\sigma_l\in\Sigma_l,\sigma_f\in\Sigma_f} r_l(\sigma_l)r_f(\sigma_f)g_i(\sigma_l,\sigma_f).
\] \section{Stackelberg Setting}
We will focus on settings where the leader first commits to a strategy that the
follower observes. The follower then plays a best response to the leader
strategy. A \emph{strong Stackelberg equilibrium} (SSE) is a pair of strategies
$r_l,r_f$ such that $r_f$ is a best response to $r_l$ and $r_l$ is a solution to
the optimization problem of maximizing $u(r_l,r_f)$ over $r_l$ and $r_f$,
subject to the constraint that $r_f$ is a best response to $r_l$. This
definition implies the common assumption that the follower breaks ties in favor
of Player
$l$~\citep{Tambe11:Security,Conitzer06:Computinga,Paruchuri08:Playing}. A
\emph{weak Stackelberg equilibrium} assumes minimization over the set of optimal
best responses.  
\section{Limited-Lookahead Model}
We will also consider a limited-lookahead variant of EFGs. There has
been a significant amount of work on limited lookahead in perfect-information
games (such as chess and checkers) in the AI community. Modeling limited
lookahead in imperfect-information games (that have information sets) is more
intricate. A model for that was presented recently~\cite{Kroer15:Limited}, and
we use that model. In that model, the follower can only look ahead $k$ steps. He
uses a \emph{node-evaluation function} $\tilde{u} : H \rightarrow \mathbb{R}$
that associates a heuristic utility with any node in the game tree. At any
information set $I\in \mathcal{I}_f$, the follower has a set of nodes
$\tilde{H}_I\subset H$ called the \emph{lookahead frontier}. When choosing his
action at information set $I$, the follower chooses an action that maximizes the
expected value of $\tilde{u}$, assuming that they choose actions so as to
maximize $\tilde{u}$ at any follower information sets reached before
$\tilde{H}_I$. We let $g_I(\sigma_l,\sigma_f)$ be the expected value over
lookahead-frontier nodes according to the node-evaluation function (analogous to
$g_i$ for the setting without limited lookahead). We assume that for any
information set $I'\in \mathcal{I}_f$ that comes after $I$, all the nodes of
$I'$ are entirely contained in the set of nodes that precede $\tilde{H}_I$, or
entirely disjoint with the set of preceding nodes (this is in order to avoid any
information sets belonging to the follower being only partially contained in the
hypothetical decision making under $I$). We let the set of information sets that
come after $I$ such that their nodes are all preceding $\tilde{H}_I$ be denoted
by $\cI_I$. We $\Sigma_f^I \subseteq \Sigma_f$ denote the set of all sequences
beneath a given information set $I$ that are within the lookahead frontier.

In the prior paper on limited lookahead in imperfect-information games~\cite{Kroer15:Limited} it was assumed that the leader knows the follower's node evaluation function exactly. That seems quite unrealistic. Therefore, we will extend the work to the case where the leader has uncertainty about the follower's node evaluation function.
  
\section{Best Responses and how to Compute Them}

Our solution concept will depend on the notion of a \emph{best response} for the
follower. For a given leader strategy $r_l$ and utility function $u_f\in U_f$,
the set of best responses is
\begin{align*}
BR(r_l,u_f) = \{r_f : u_f(r_l,r_f) = \max_{r_f'} u_f(r_l, r_f')\}.
\end{align*}

Given a strategy $r_l$ for the leader and a utility function $u_f$, the value of each information set can be
computed with the following feasibility program (this holds outside of a
leader-follower setting as well):

\noindent\begin{minipage}{8.4cm}
    \begin{equation*}
  	  v_{\infoset_f(\sigma_f)} = s_{\sigma_f} + \sum_{\mathclap{\substack{I'\in \cI_f\nonumber\\\seq_f(I')=\sigma_f}}} v_{I'} + \sum_{\sigma_l \in \Sigma} r_l(\sigma_l)g_f^{u_f}(\sigma_l,\sigma_f)
    \end{equation*}\vspace{-5mm}
    \begin{equation}
      \label{eq:infoset_value_constraint}
	  \hspace{4.5cm}\forall I\in \cI_f, \sigma_f=\seq_f(I)
    \end{equation}
    \vspace{-1mm}
    \begin{equation}
      0 \leq s_{\sigma_f} \leq M (1-b_f(\sigma_f)) \hspace{2.3cm} \forall \sigma_f\in \Sigma_f
    \end{equation}
	\begin{equation}
	  \textstyle\sum_{a\in A(I)}b_f(\sigma a) = 1 \hspace{1.5cm} \forall I\in \cI_f, \sigma_f=\seq_f(I)
	\end{equation}
	\begin{equation} 
	  b_f(\sigma_f) \in \{0,1\} \hspace{4.cm} \sigma_f \in \Sigma_f \label{eq:br_01}
	\end{equation}
\end{minipage}\vspace{2mm}

\noindent The variables $v_I$ represent the value of a given information set $I$, $b_f(\sigma_f)$ represents whether $\sigma_f$ is a best response at its
respective information set, and $s_{\sigma_f}$ represents how much less utility
the follower gets by following the sequence $\sigma_f$ rather than the optimal
action at $\infoset(\sigma_f)$.
  It is easy to show via induction that
  the feasibility MIP given in
  equations~(\ref{eq:infoset_value_constraint}-\ref{eq:br_01}) computes a best
  response to $r_l$ and the variables $v_I$ represent the values of
  information sets $I$ when best-responding to $r_l$:
  For the base case of an information set with no future
  information sets belonging to the follower, disregarding $s_{\sigma_f}$, the
  RHS of \eqref{eq:infoset_value_constraint} clearly represents the value of
  choosing $\sigma_f$ at the information set. Now, since all $s_{\sigma_f}$ are
  nonnegative and \eqref{eq:infoset_value_constraint} is an equality, it follows
  that $v_I$ upper bounds the value of each individual sequence at $I$. But
  since $s_{\sigma_f}=0$ for some $\sigma_f$, it must be an equality for said
  $\sigma_f$. Thus $v_I$ upper bounds the value of all sequences at $I$, but is
  also equal to the value of some sequence, and therefore it represents the
  value when best responding. Applying the inductive hypothesis to any
  information set $I$ that has future information sets belonging to the follower
  reduces the expression for $v_I$ to one that is equivalent to the base case.
 
\section{Extension to Uncertainty about the Opponent}
We now extend the EFG model to incorporate uncertainty about the
follower's utility function. We will take a robustness approach, where we care
about the worst-case instantiation of the uncertainty set $U_f$. For
limited-lookahead EFGs we will analogously consider uncertainty over the
node-evaluation function.

Due to the uncertainty (represented by the uncertainty set $U_f$), defining a Stackelberg
equilibrium is not straightforward.
We take the perspective that a robust Stackelberg solution is a strategy for the
leader that maximizes the leader utility in the worst-case instantiation of
$U_f$:
\begin{definition}
  A \emph{robust strong Stackelberg solution} (RSSS) is a realization
  plan $r_l$ such that
  \[
    r_l \in \arg \max_{r_l' \in R_l} \quad \inf_{u_f\in U_f} \quad \max_{r'_f \in
      BR(r_l,u_f)} u_l(r_l,r'_f).
  \]
  \end{definition}
The robustness is represented by the minimization over $U_f$. Intuitively, if
the actual instantiation of $u_f$ does not take on the minimizer over $U_f$,
the leader can only receive better utility, so we are computing the maximin
utility against the robustness. Typically one is interested in finding
an RSSS strategy for the leader, but we
nonetheless define the entire equilibrium concept as well:
\begin{definition}
  A \emph{robust strong Stackelberg equilibrium} (RSSE) is a realization plan
  $r_l$ and a (potentially uncountably large) set of realization plans
  $\{r_f^{u_f} : \forall u_f\in U_f\}$ such that $r_l$ is an RSSS and $r_f^{u_f}
  \in BR(r_l,u_f)$ for all $u_f\in U_f$.
\end{definition}
Whether an RSSE is even practical to represent is highly dependent on the
structure of the specific game and uncertainty sets at hand, as it would
frequently need to be represented parametrically. On the other hand, once we
have $r_l$, the best response for a specific $u_f$ can easily be computed. One
method for doing this is to first compute the follower value $u^*$ under $u_f$
when best responding to $r_l$ (e.g., via a single tree traversal), and then
solving the \emph{linear program (LP)} that consists of maximizing the leader's utility over
the set of follower strategies that achieve $u^*$ (this can be done by adding a
single constraint to the sequence-form best-response LP given by \citet{Stengel96:Efficient}).

One might consider applying the robustness after the follower chooses her
strategy (in a sense, swapping the inner max and min). In this case, we cannot
represent this as a minimization on the inside since the set of best responses
is defined with respect to the choice of $u_f$. Arguably the most natural way to
apply robustness after the best response of the follower would be to ask for a
pair of strategies $r_l,r_f$ such that $r_f$ is a best response no matter the
instantiation of $u_f$. This definition of robustness would allow us to apply
standard robust optimization techniques to any Stackelberg MIP. However, this definition has several drawbacks.
First, if we are applying a robust model, we are often interested in
maximizing our worst-case utility. By applying robustness after choosing $r_f$,
we would not be doing that, but instead would be maximizing utility subject to the
constraint that we want to be sure what the follower response is.
Second, a robust Stackelberg equilibrium defined that way would not
necessarily exist: if there is overlap between the range of possible
utilities associated with a pair of actions at some information set, there would
be no way to guarantee that a single action will always be a best response.

% \end{proof} 
 \section{MIP for Full-Certainty Setting}

We now give a MIP for computing a Stackelberg equilibrium in a game where the
follower has limited lookahead. 

\noindent\begin{minipage}{8.4cm}
	\begin{equation}
	  \label{eq:mip_obj_full_certainty}
	  \max_{p,r,v,s} \sum_{z\in Z} p(z)u_l(z)\mathcal{C}(z)
	\end{equation}
	\begin{equation*}
	  v_{I} = s_{\sigma_f} + \sum_{\mathllap{I'\in \cI_f}:\seq_f(I')\mathrlap{=\sigma_f}} v_{I,I'} + \sum_{\sigma_l \in \Sigma_l} r_l(\sigma_l)g_I(\sigma_l,\sigma_f)
	\end{equation*}\vspace{-3mm}
	\begin{equation}
	  \hspace{4cm}\label{eq:infoset_slack_full_certainty} \forall \sigma_f\in \Sigma_f, I=\infoset_f(\sigma_f)
	\end{equation}
	\vspace{-1mm}
	\begin{equation*}
	  v_{I,\infoset(\sigma_f)} = s_{\sigma_f}^{I} + \sum_{\mathllap{I'\in \cI_f}:\seq_f\mathrlap{(I')=\sigma_f}} v_{I,I'} + \sum_{\sigma_l \in \Sigma_l} r_l(\sigma_l)g_I(\sigma_l,\sigma_f)
	\end{equation*}\vspace{-3mm}
	\begin{equation}
	  \hspace{5cm}\label{eq:mip_lookahead_constraint_full_certainty}
	  \forall I\in \mathcal{I}, \sigma_f\in \Sigma_f^I 
	\end{equation}
	\vspace{-1mm}
	\begin{equation}
	  \label{eq:sequence_constraints_full_certainty} r_i(\emptyset)=1 \hspace{4.3cm} \forall i\in \{l,f\} \hspace{.3cm}
	\end{equation}
	\begin{equation}\textstyle
    \label{eq:sequence_constraints_full_certainty2}
	  r_i(\sigma_i)=\sum_{a\in A(I_i)}r_i(\sigma_ia) \hspace{.7cm} \forall i\in \{l,f\}, I_i\in \cI_i \hspace{.3cm}
	\end{equation}
	\begin{equation}
	  \label{eq:slack_constraints_full_certainty}
	  0\leq s_{\sigma_f} \leq \left( 1-r_f(\sigma_f) \right)M \hspace{2cm} \forall \sigma_f\in \Sigma_f 
	\end{equation}
	\begin{equation}
	0\leq s_{\sigma_f}^I \leq \left( 1-r_f^I(\sigma_f) \right)M \hspace{.8cm} \forall I\in \cI_f, \sigma_f^I\in \Sigma_f^I
	\end{equation}
	\begin{equation}
	  r_f(\sigma_f) \in \{0,1\} \hspace{3.6cm} \forall \sigma_f \in \Sigma_f
	\end{equation}
	\begin{equation}
	  \label{eq:best_response_full_certainty}
	  r^I_f(\sigma_f) \in \{0,1\} \hspace{2.6cm} \forall I\in \cI_f\sigma_f \in \Sigma_f^I
	\end{equation}
	\begin{equation}
	  \label{eq:leaf_prob_full_certainty}
	  0 \leq p(z) \leq r_i(\seq_i(z)) \hspace{1.7cm} \forall i\in \{l,f\}, z\in Z
	\end{equation}
	\begin{equation}
	  \label{eq:leaf_prob_sum_full_certainty}
	  1 = \textstyle\sum_{z\in Z} p(z) \mathcal{C}(z)\hspace{4.5cm}
	\end{equation}
	\begin{equation} 
	  \label{eq:full_certainty_mip_last}
	  0 \leq r_l(\sigma_l) \leq 1 \hspace{4.0cm} \forall \sigma_l \in \Sigma_l
	\end{equation}
\end{minipage}\vspace{2mm}
This MIP is an extension of the MIP given by \citet{Bosansky15:Sequence} to
the limited-lookahead setting of \citet{Kroer15:Limited}.
Eq. \eqref{eq:mip_obj_full_certainty} is the expected leader value over leaf nodes. Equations \eqref{eq:infoset_slack_full_certainty} to
\eqref{eq:best_response_full_certainty} set up best-response constraints for
each follower information set, as well as for each pair of information sets
$I,I'$ such that $I'\in \mathcal{I}_I$ (these constraints are completely
analogous to (\ref{eq:infoset_value_constraint})-(\ref{eq:br_01}) except that
the constraints involving $v_{I,I'}$ must be set up for each $I$ in order to
represent best responses when applying the lookahead evaluation function at $I$).
Equations \eqref{eq:leaf_prob_full_certainty} and~\eqref{eq:leaf_prob_sum_full_certainty}
ensure that the probabilities over leaves are correct. Finally
\eqref{eq:sequence_constraints_full_certainty}, \eqref{eq:sequence_constraints_full_certainty2}, and
\eqref{eq:full_certainty_mip_last} ensure that $r_l$ is a valid leader strategy.
 
\section{MIP with Uncertainty about Follower Payoff}
We now move to the computation of RSSS for the setting with uncertainty about
follower payoff but no limited lookahead. We will consider a particular class of
uncertainty functions: interval uncertainty on each leaf payoff. More
concretely, the uncertainty set will be
\[
  U_f = \{ u_f : u_f(h)\in [L(h),U(h)], \forall h\in Z\},
\]
where $L(h),U(h)$ are given upper and lower bounds on the interval that the
payoff for leaf node $h$ must be chosen from.

One issue that now arises is that we may not be able to make a single action
optimal: if the maximum-to-minimum utility intervals for two sequences are
guaranteed to overlap we cannot make either sequence the optimal choice for the
follower player. Instead, we allow choosing both sequences, and we then
assume that the leader player receives the minimum over the two.
Intuitively, this can be thought of as a zero-sum game played within the space
of actions made optimal for the follower player (a similar technique was
used in~\citet{Kroer15:Limited}). More generally, we may have $k>1$
actions at a given information set that can all be made optimal under various
instantiations of the utility function. We now introduce a set-valued function
that, under some given strategy for the follower $r_f$, returns the set of
actions at a given follower information set that can be made optimal under some
instantiation of the utility function, given the tie-breaking rule,
\[
  A_I(r_f) = \{ a\in A_I : \nexists\ a'\in A_I,\ v_I^L(a') \geq v_I^U(a) \}.
\]
For any $a\in A_I(r_f)$, the minimization over the uncertainty can choose an
instantiation making $a$ the only best-response action at $I$. Conversely, for
$a\notin A_I(r_f)$, even if the utility function is chosen to maximize the value
of action $a$, there exists some other action $a'$ whose worst-case instantiation is
at least as good; if $a$ leads to better leader utility than $a'$ then the
minimization over utility functions will not allow them to be tied, and if $a$
leads to worse utility than $a'$, then even if a utility function causing a tie
is chosen, the best-response tie-breaking in favor of the leader means that $a'$
will be chosen. Thus, $a'$ (or some other action) is always chosen over $a$.\footnote{Here we rely on the
  assumption that every action has a strict inequality $v_I^U(a) > v_I^L(a)$.
  Without this assumption our MIP still works, but the math becomes
  more cumbersome.}

The function $A_I(r_f)$ is illustrated in
Figure~\ref{fig:active_actions}. The general intuition can be seen from the
figure: the dotted line denotes the split between potentially optimal actions
(black bars) and actions that cannot be made optimal through any
utility-function choice (opaque bars). Note that the dotted line is touched by
interval end-points from both sets: this means that the two actions could be
tied, but the lower-value would never be chosen, since it is either worse for
the leader, in which case the tie-breaking does not choose it even in case of a
tie, or if it is better then the minimization over the intervals will break the
tie and make it inoptimal.

The intuition behind our robust MIP consists of
three components: 1) the best-response
feasibility MIP described in \eqref{eq:infoset_value_constraint}-\eqref{eq:br_01}, instantiated independently for both the
set of maximal and minimal valuation functions, 2) a set of constraints for computing the
set $A_I(r_f)$ for a given $r_f$ via  
best-response values for the maximal and minimal  utility functions, and 3) a
minimization similar to the dual best-response LP from the standard
sequence-form LP~\citep{Stengel96:Efficient}.
\begin{figure}
  \centering
  \scalebox{0.7}{
  \begin{tikzpicture}
    \draw[gray,dashed] (-0*1.21, -.2*1.2) -- (-0*1.21, 1.8);
    \path[fill=black] (.3*1.21, 0*1.21) rectangle (1.3*1.21, .15*1.21);
    \path[fill=black] (-0*1.21, .3*1.21) rectangle (1.0*1.21, .45*1.21);
    \path[fill=black] (-1.0*1.21, .6*1.21) rectangle (1.6*1.21, .75*1.21);
    \path[draw, fill=white] (-2.2*1.21, .9*1.21) rectangle (-.85*1.21, 1.05*1.21);
    \path[draw, fill=white] (-1.9*1.21, 1.2*1.21) rectangle (-.0*1.21, 1.35*1.21);
    \draw[->] (-2.6*1.2, -.2*1.2) -- (1.9*1.2, -.2*1.2);
    \node at (2.4*1.2, -.2*1.2) {Utility};
  \end{tikzpicture}
  }
  \caption{A set of action value-uncertainty intervals.}
  \label{fig:active_actions}
\end{figure}
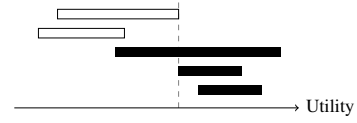

In the robust MIP given below, $g_f^U,g_f^L$ are the functions giving the
expected value over leaf nodes consistent with a pair of sequences, when every
node has its payoff set to the maximal ($g_f^U$) and minimal ($g_f^U$) payoff, respectively.

\noindent\begin{minipage}{8.4cm}
	\begin{equation}
	  \label{eq:mip_obj_robust}
	  \min_{r,v,s,y} y_0 
	\end{equation}
	\begin{equation*}
	  y_{\infoset_f(\sigma_f)} \geq\sum_{\mathclap{\substack{I'\in \cI_f\\\seq_f(I')=\sigma_f}}} y_{I'} - \sum_{\sigma_l \in \Sigma_l} g_l(\sigma_l,\sigma_f) r_l(\sigma_l) - M(1-r_f(\sigma_f))
	\end{equation*}\vspace{-6mm}
	\begin{equation}
	  \hspace{6cm}\forall \sigma_f \in \Sigma_f \label{eq:zerosum_dual_constraint_robust}
	\end{equation}
	\vspace{-1mm}
				\begin{equation*}
	  v_{\infoset_f(\sigma_f)}^q = s^q_{\sigma_f} + \sum_{\mathclap{\substack{I'\in \cI_f\\\seq_f(I')=\sigma_f}}} v^q_{I'} + \sum_{\sigma_l \in \Sigma} r_l(\sigma_l)g^q_f(\sigma_l,\sigma_f)
	\end{equation*}\vspace{-5mm}
	\begin{equation}
	\hspace{4.27cm}\forall \sigma_f \in \Sigma_f, q\in \{U,L\}  \label{eq:infoset_value_constraint_robust}
	\end{equation}
	\vspace{-1mm}
	\begin{equation}
	  \label{eq:slack_constraints_robust}
	0 \leq s^q_{\sigma_f} \leq M (1-b_f^q(\sigma_f)) \hspace{.5cm} \forall \sigma_f\in \Sigma_f, q\in \{U,L\}
	\end{equation}
	\begin{equation}
	  \textstyle\sum_{{a\in A(I)}}b_f^q(\seq_f(I) a) = 1 \hspace{.8cm} \forall q\in \{U,L\}, I\in \cI_f
	\end{equation}
	\begin{equation}
	  \label{eq:robust_mip_last_br}
	  b_f^q(\sigma_f) \in \{0,1\} \hspace{1.9cm} \forall \sigma_f \in \Sigma_f, q\in \{U,L\}
	\end{equation}
	\begin{equation}
	\label{eq:optimal_actions_robust}
	  v_I^U - s_{\sigma_f}^U \geq v_I^L - M(1-r_f(\sigma_f)) \hspace{1.5cm} \forall \sigma_f \in \Sigma_f 
	\end{equation}
	\begin{equation}
	  \label{eq:inoptimal_actions_robust}
	  v_I^U - s_{\sigma_f}^U \leq v_I^L + Mr_f(\sigma_f) \hspace{2cm} \forall \sigma_f \in \Sigma_f
	\end{equation}
	\begin{equation}
	  \label{eq:empty_sequence_constraint_robust}
	  r_i(\emptyset)=1 \hspace{4.5cm} \forall i\in \{l,f\}
	\end{equation}
	\begin{equation}
	  \label{eq:sequence_constraints_robust}
	  r_l(\sigma)=\sum_{a\in A(I)}r_l(\sigma a) \hspace{1.6cm} \forall I\in \cI_l, \sigma=\seq_l(I)
	\end{equation}
	\begin{equation}
	  r_f(\sigma)\leq \sum_{a\in A(I)}r_f(\sigma a) \hspace{1.3cm} \forall I\in \cI_f, \sigma=\seq_f(I)
	\end{equation}
	\begin{equation}
	  r_f(\sigma_f) \in \{0,1\} \hspace{3.7cm} \forall \sigma_f \in \Sigma_f
	\end{equation}
	\begin{equation}
	  0 \leq r_l(\sigma_l) \leq 1 \hspace{4cm} \forall \sigma_l \in \Sigma_l \label{eq:mip_probability_robust}
	\end{equation}
\end{minipage}
\vspace{1mm}

\noindent Equations \eqref{eq:mip_obj_robust} and \eqref{eq:zerosum_dual_constraint_robust}
implement the minimization over the set of potentially optimal actions
$A_I(r_l)$ at a given information set $I$.
Equations \eqref{eq:infoset_value_constraint_robust} to \eqref{eq:robust_mip_last_br}
ensure that $v_I^U,v_I^L$ represent the correct value of each information set
under the maximal and minimal utility function.
Equations \eqref{eq:optimal_actions_robust} and \eqref{eq:inoptimal_actions_robust} ensure
that actions in or not in $A_I(r_l)$ can potentially be made optimal
\eqref{eq:optimal_actions_robust} or cannot be made
optimal~\eqref{eq:inoptimal_actions_robust}.
Equations \eqref{eq:empty_sequence_constraint_robust} to \eqref{eq:mip_probability_robust}
ensure that $r_l$ is a valid sequence-form leader strategy and that one more
pure strategies are active for the follower.
We prove that this MIP computes a RSSS. Due to space constraints the result is
shown in the appendix.

\subsection{MIP for Limited-Lookahead Interval Uncertainty}

We also present an extension of the full-certainty MIP for limited lookahead to
a setting with uncertainty about the limited-lookahead 
node-evaluation function. That MIP joins the ideas from both the full-certainty
MIP with limited
lookahead~(\eqref{eq:mip_obj_full_certainty}-\eqref{eq:full_certainty_mip_last})
and the robust MIP (\eqref{eq:mip_obj_robust}-\eqref{eq:mip_probability_robust})
and is thus the most comprehensive, but it combines the novel ideas from the
former two MIPs in a fairly straightforward way. Due to limited space we present
the MIP in the appendix.
 
\section{Experiments}
Using our MIPs presented in the previous section we investigated the scalability
and qualitative properties of RSSS solutions. We experimented with three kinds of EFG:
Kuhn poker (Kuhn)~\cite{Kuhn50:Simplified}, a 2-card poker variant (2-card), and a parameterized security-inspired
search game (Search). The
search game is similar to games considered by \citet{Bosansky14:Exact} and
\citet{Bosansky15:Sequence}).

Kuhn consists of a three-card deck: king, queen, and jack. Each player
first has to put a payment of 1 into the pot. Each player is then dealt one of the three cards, and the third is put
aside unseen. A single round of betting then occurs (with betting parameter $p=1$, explained below).

In 2-card, the deck consists of two kings and two jacks. Each player
first has to put a payment of 1 into the pot.
A private card is dealt to each, followed by a betting round (with betting parameter $p=2$), then a public card is dealt, followed by another betting round (with $p=4$). 

In both games, each round of betting goes as follows:
\begin{compactitem}
\item Player $1$ can check or bet $p$.
  \begin{compactitem}
  \item If Player $1$ checks Player $2$ can check or raise $p$.
    \begin{compactitem}
      \item If Player $2$ checks the betting round ends.
      \item If Player $2$ raises Player $1$ can fold or call.
        \begin{compactitem}
          \item If Player $1$ folds Player $2$ takes the pot.
          \item If Player $1$ calls the betting round ends.
          \end{compactitem}
        \end{compactitem}
      \item If Player $1$ raises Player $2$ can fold or call.
        \begin{compactitem}
        \item If Player $2$ folds Player $1$ takes the pot.
        \item If Player $2$ calls the betting round ends.
        \end{compactitem}
      \end{compactitem}
\end{compactitem}

If no player has folded, a showdown occurs.
In Kuhn poker, the player with the higher card wins in a showdown. In 2-card, showdowns have two possible
outcomes: one player has a pair, or both players have the same private
card. For the former, the player with the pair wins the pot. For the
latter the pot is split.

Kuhn poker has $55$ nodes in the game tree and $13$ sequences per player. The 2-card game tree has $199$ nodes, and $57$ sequences per player.

The search game is played on the graph shown in Figure~\ref{fig:search_game}.
It is a simultaneous-move game (which can be modeled as a turn-taking EFG with
appropriately chosen information sets). The leader controls two patrols that can
each move within their respective shaded areas (labeled P1 and P2), and at each
time step the controller chooses a move for both patrols. The follower is always
at a single node on the graph, initially the leftmost node labeled $S$ and can
move freely to any adjacent node (except at patrolled nodes,
the follower cannot move from a patrolled node to another patrolled node). The
follower can also choose to wait in
place for a time step in order to clean up their traces. 
If a patrol visits a
node that was previously visited by the follower, and the follower did not wait
to clean up their traces, they can see that the follower was there. If the
follower reaches any of the rightmost nodes they received the respective payoff
at the node ($5$, $10$, or $3$, respectively). If the follower and any patrol are
 on the same node at any time step, the follower is captured,
which leads to a payoff of $0$ for the follower and a payoff of $1$ for the
leader. Finally, the game times out after $k$ simultaneous moves, in which case
the leader receives payoff $0$ and the follower receives $-\infty$ (because we
are interested in games where the follower attempts to reach an end node). We
consider games with $k$ being $5$ and $6$. We will denote these by Search-5 and
Search-6. Search-5 (Search-6) has 87,927 (194,105) nodes and 11,830 and 69
(68,951 and 78) leader and follower sequences.
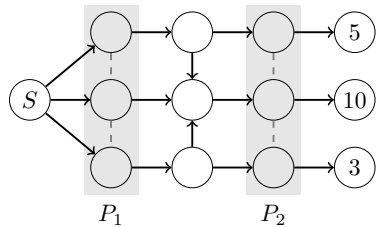
\begin{figure}
  \centering
  \scalebox{0.9}{
\begin{tikzpicture}
\path[fill=black!10!white] (.8, -1.4) rectangle (1.6, 1.4);
\path[fill=black!10!white] (3.2, -1.4) rectangle (4.0, 1.4);
\node at (1.2, -1.7) {$P_1$};
\node at (3.6, -1.7) {$P_2$};

  \node[draw, circle, minimum width=.6cm, inner sep=0] (A) at (0, 0) {$S$};
  \node[draw, circle, minimum width=.6cm] (B) at (1.2, 1) {};
  \node[draw, circle, minimum width=.6cm] (C) at (1.2, 0) {};
  \node[draw, circle, minimum width=.6cm] (D) at (1.2, -1) {};
    \node[draw, circle, minimum width=.6cm] (E) at (2.4, 1) {};
    \node[draw, circle, minimum width=.6cm] (F) at (2.4, 0) {};
    \node[draw, circle, minimum width=.6cm] (G) at (2.4, -1) {};
      \node[draw, circle, minimum width=.6cm] (H) at (3.6, 1) {};
      \node[draw, circle, minimum width=.6cm] (I) at (3.6, 0) {};
      \node[draw, circle, minimum width=.6cm] (J) at (3.6, -1) {};
            \node[draw, circle, minimum width=.6cm,inner sep=0] (K) at (4.8, 1) {$5$};
            \node[draw, circle, minimum width=.6cm,inner sep=0] (L) at (4.8, 0) {$10$};
            \node[draw, circle, minimum width=.6cm,inner sep=0] (M) at (4.8, -1) {$3$};
            
\draw[thick,->] (A) edge (B);
\draw[thick,->] (A) edge (C);
\draw[thick,->] (A) edge (D);
\draw[thick,->] (B) edge (E);
\draw[thick,->] (C) edge (F);
\draw[thick,->] (D) edge (G);
\draw[thick,->] (E) edge (F);
\draw[thick,->] (G) edge (F);
\draw[thick,->] (E) edge (H);
\draw[thick,->] (F) edge (I);
\draw[thick,->] (G) edge (J);
\draw[thick,->] (H) edge (K);
\draw[thick,->] (I) edge (L);
\draw[thick,->] (J) edge (M);

\draw[thick,gray,dashed] (B) edge (C);
\draw[thick,gray,dashed] (C) edge (D);

\draw[thick,gray,dashed] (H) edge (I);
\draw[thick,gray,dashed] (I) edge (J);
\end{tikzpicture}
   }
  \caption{The graph on which the search game is played.}
  \label{fig:search_game}
\end{figure}

All experiments were conducted using Gurobi 7.5.1 to solve MIPs, on a cluster
with 8 Intel Xeon E5607 2.2Ghz cores and 47 GB RAM per experiment.

In the first set of experiments we investigate the impact on runtime caused by
uncertainty intervals in each of the four games, without
considering limited lookahead. We compare the 
MIP by \citet{Bosansky15:Sequence}
(B\&C) for the full-certainty
setting to our robust MIP
(\eqref{eq:mip_obj_robust}-\eqref{eq:mip_probability_robust}) with an
uncertainty interval of diameter $d$ at each node in the game, for $6$ different
values of $d$. The results are shown in Table~\ref{tab:exp_runtime}.
Interestingly, our robust MIP with interval $0$ is significantly faster than the
B\&C 
MIP for the Search games. (We do not specialize our robust MIP to the
full-certainty setting but instead let Gurobi presolve away most redundant
variables and constraints. One could easily specialize it and potentially make
it even faster.) Once we add uncertainty, the MIP gets
harder to solve, with the runtime increasing for larger uncertainty intervals---except for the largest uncertainty interval where the problem starts to get
easier again.

\begin{table}
  \centering\includegraphics[scale=1]{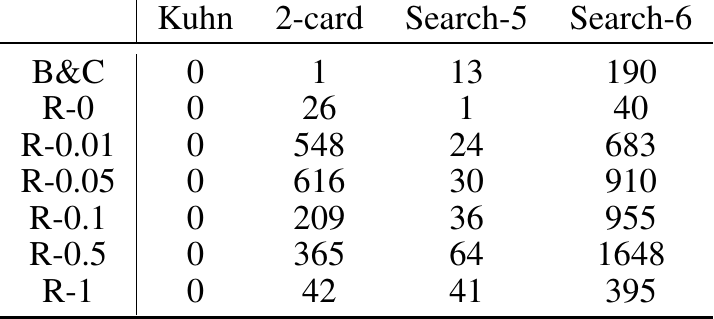}
  \caption{Runtime experiments for the MIP by \citet{Bosansky15:Sequence} (B\&
    C) and our robust Stackelberg MIP for increasing uniform uncertainty
    intervals (R-c where c is the interval radius). All runtimes are in
    seconds.}
  \label{tab:exp_runtime}
\end{table}

In the second set of experiments, we investigate the cost of computing an RSSS
against a follower utility function that is different from the one actually
employed by the follower. These experiments were conducted on the Search-5
game. On Search-6 it would take prohibitively long to conduct all the experiments, and
the experiments would not be interesting on Kuhn and 2-card because they are zero-sum games 
(the leader will end up getting the value of the game as long as the correct
utility function is contained in the uncertainty intervals). The setup is as
follows. We use our robust MIP to compute a leader strategy for the original
payoffs in Search-5. We instantiate the MIP with several different
uncertainty-interval widths (given in the leftmost column in
Table~\ref{tab:exp_misspecification_cost}). 
For each leader strategy, we then conduct a grid search
over triplets of numbers in $\{ \pm 0.1,\pm 0.5,\pm 1,\pm 2,\pm 3\}^3$, where
the three numbers correspond to a change in utility being added to each of the
three rightmost payoff nodes in Figure~\ref{fig:search_game}. 
For each payoff
change, we compute the follower's best response (breaking ties in favor of the
leader) to the leader strategy under the new game and the resulting leader
utility. The second column in Table~\ref{tab:exp_misspecification_cost} (EV) denotes the value that
the leader is expected to get if he were solving the correct game. The following
three columns, labeled $\leq 1, \leq 2, \leq 3$, show the worst utility achieved
by the leader when the grid search is restricted to payoff changes of at most 1,
2, and 3, respectively. For example, in the case $\leq 1$ we only do the grid
search over  $\{ \pm 0.1,\pm 0.5,\pm 1\}^3$
The experiment shows that when uncertainty is not taken into account, all amounts
of perturbation leads to a large decrease in leader utility. Conversely, taking
uncertainty into account leads to much better utility in almost every case.

\npdecimalsign{.}
\nprounddigits{2}
\begin{table}
  \centering\includegraphics[scale=1]{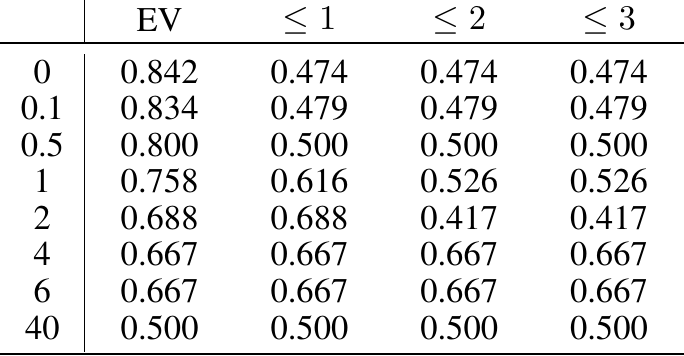}
  \caption{Leader utility when maximizing  utility against an incorrect utility function.
    Each row corresponds to a different size of uncertainty interval used for
    computing the leader strategy (interval size is given in the leftmost
    column). The columns are ordered in increasing amounts of incorrectness
    allowed in the follower utility function.}
  \label{tab:exp_misspecification_cost}
\end{table}
\npnoround

In the third set of experiments, we investigate the cost to the leader from
having to take uncertainty into account against a limited-lookahead follower. We
perform this experiment on Kuhn and 2-card, both zero-sum games, which allows us
to apply the same node-evaluation scheme as in \citet{Kroer15:Limited}. 
In order to construct the limited-lookahead evaluation function, we first
compute a Nash equilibrium of the game. We then recursively define the value of
each node to be the weighted sum over the values of nodes beneath it, where the
weights are the probabilities of each action in the Nash equilibrium, and then
add Gaussian noise to the computed value (we do not add any noise to leaf
nodes). Since the value of a node is based on the noisy value of nodes beneath
it, the farther away from leaf nodes a node is, the noisier the estimate of the
node's value (from Nash equilibrium) is. We then use our robust
limited-lookahead MIP to solve the limited-lookahead game resulting from having
the follower apply this node-evaluation function. We consider lookahead depths
of $1$ and $2$. The results for 2-card are shown in
Table~\ref{tab:lookahead_2card} and the results for Kuhn are shown in
Table~\ref{tab:lookahead_kuhn}. The different rows in the tables correspond to
varying standard deviations in the Gaussian noise, and columns correspond to
increasing sizes of uncertainty intervals. For all games, lookahead depths, and
noise levels, we see that the amount that the leader can exploit the follower
goes down as uncertainty intervals get larger. However, we also see that for
most noise amounts, some amount of robustness can be added without losing
substantial leader utility. Coupled with our results from the second set of
experiments, which showed that uncertainty intervals are necessary if there is
mis-specification in the model, this suggests that uncertainty intervals can
lead to substantially more robust outcomes, potentially at a small cost to
optimality even if the initial model turns out to be correct.

\begin{table}
  \centering\includegraphics[scale=.95]{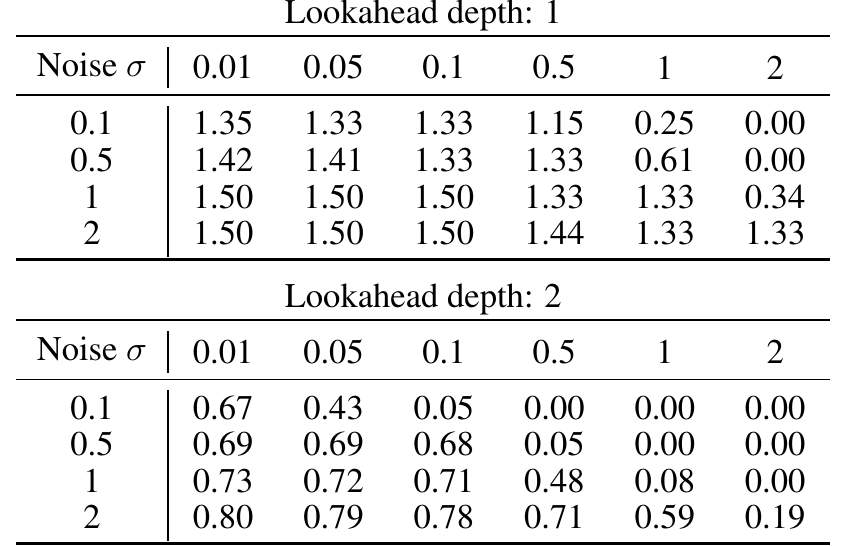}
  \caption{Limited-lookahead with depth 1 and 2 in 2-card.
    }
  \label{tab:lookahead_2card}
\end{table}

\begin{table}
  \centering\includegraphics[scale=.95]{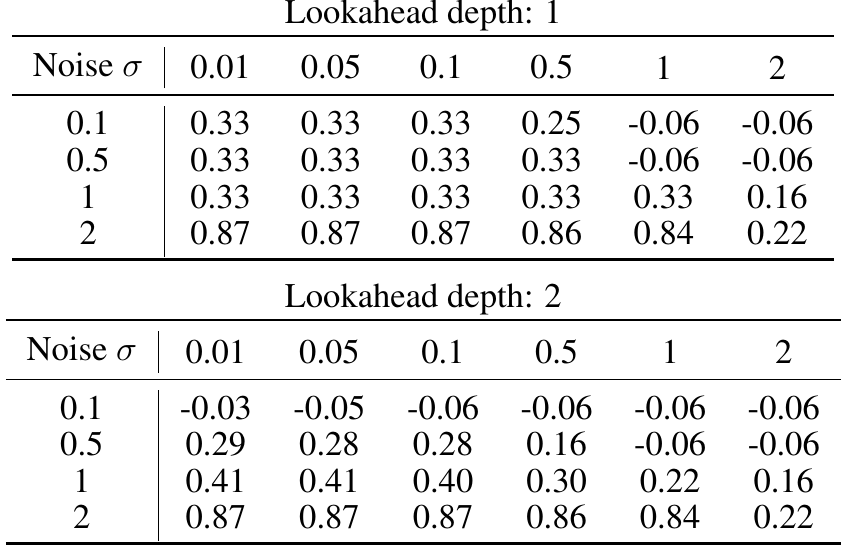}
  \caption{Limited-lookahead with depth 1 and 2 in Kuhn.
    }
  \label{tab:lookahead_kuhn}
\end{table}
 
\section{Discussion}
While we showed that our technique scales to medium-size games, in practice we
would often like to scale to even larger games. The iterative LP-based approach
of \citet{Cermak16:Using} could potentially be extended to the robust setting.
Likewise, abstraction methods have dramatically increased the scalability of
Nash equilibrium finding in EFGs
(e.g.,~\citep{Gilpin07:Lossless,Lanctot12:No,Kroer14:Extensive,Kroer16:Imperfect,Brown15:Hierarchical})
and could potentially be adapted to the robust Stackelberg setting as well. This
could be done while giving guarantees on follower behavior by only abstracting
the strategy space of the leader.

\section{Acknowledgements}
This material is based on work supported by the National Science Foundation
under grants IIS-1718457, IIS-1617590, and CCF-1733556, and the ARO under award
W911NF-17-1-0082. Christian Kroer is also sponsored by a
Facebook Fellowship.

\bibliographystyle{aaai}      \bibliography{../../dairefs/dairefs}   
\appendix

\section{Proof of Interval Uncertainty MIP Correctness}
\begin{theorem}
  \label{th:robust_mip_correct}
  The interval uncertainty MIP computes a robust Stackelberg equilibrium for an
  EFG with interval uncertainty on each leaf payoff for the follower.
\end{theorem}
\begin{proof}
  First we show that a robust Stackelberg equilibrium corresponds to a solution
  to the MIP. Let $r_l,r_f$ be a robust Stackelberg equilibrium. Without loss of
  generality, assume that $r_f$ is a pure strategy (for any mixed-strategy best
  response, by the assumption of tie-breaking in favor of the leader, a pure
  best response exists yielding the same utility for the leader). Set all MIP
  variables $r_l(\sigma_l)$ according to the equilibrium strategies. For all
  $\sigma_f\in A_{\infoset(\sigma_f)}(r_l)$ set the corresponding MIP variable
  $r_f(\sigma_f) = 1$, and for all $\sigma_f\notin A_{\infoset(\sigma_f)}(r_l)$ set
  $r_f=0$. For any information set $I$, the action $a$ played in the pure
  strategy $r_f$ must be in $A_I(\sigma_l)$: if not it could not be made optimal
  by any utility function. Conversely, it must be the action providing the
  lowest utility to the leader among actions in $A_I(\sigma_l)$, or the
  minimization over utility functions would not have made it optimal. Choose $y$
  so as to minimize \eqref{eq:mip_obj_robust} subject to
  \eqref{eq:zerosum_dual_constraint_robust}. This corresponds exactly to minimizing the
  leader utility over the set of follower sequences such that $r_f(\sigma_f) =
  1$, and thus the objective is equal to the value of the RSSE. This can also be
  seen by realizing that (\ref{eq:mip_obj_robust},~\ref{eq:zerosum_dual_constraint_robust})
  correspond to the dual sequence-form best-response LP of an opponent wishing
  to minimize the leader utility over $A_I(r_l)$. Set $v_I^U,b_U,s_{\sigma_f}^U$
  and $v_I^L,b_L,s_{\sigma_f}^L$ equal to the values obtained by arbitrarily
  chosen best responses according to the maximal and minimal utility functions
  respectively. By our choices for MIP variables it is clear that
  \eqref{eq:zerosum_dual_constraint_robust},
  \eqref{eq:infoset_value_constraint_robust}, and
  \eqref{eq:empty_sequence_constraint_robust}-\eqref{eq:mip_probability_robust}
  are satisfied. For \eqref{eq:optimal_actions_robust} note that we set
  $r_f(\sigma_f) = 1$ only for variables in $A_{\infoset(\sigma_f)}(\sigma_l)$, that
  is, sequences where their upper-bound value is greater than every lower-bound
  value, and thus $v_I^U-s_{\sigma_f}^U$, which is exactly the upper-bound value
  associated with $\sigma_f$, is greater than $v_I^L$. Conversely, for
  $\sigma_f$ such that $r_f(\sigma_f) = 0$ we know that their upper-bound value
  is less than some lower-bound value, and thus $v_I^U-s_{\sigma_f}^U\leq v_I^L$.

  Now consider an optimal solution to the MIP. The leader strategy for an RSSE
  is exactly the values computed for $r_l(\sigma_l)$ for all $\sigma_l$. By
  the same logic as for the standard Stackelberg MIP, $v_I^U,v_I^L$ represent
  the information-set values according to the maximal and minimal utility
  functions for the follower~\cite{Shoham09:Multiagent,Bosansky15:Sequence}. Since $v_I^U-s_\sigma^U$ corresponds exactly to
  the information-set value associated with a given sequence $\sigma \in
  \Sigma_I$, \eqref{eq:optimal_actions_robust} implies that $r_f(\sigma)=1$ if
  and only if $v_I^U(\sigma)\geq v_I^L(\sigma')$ for all $\sigma'\in \Sigma_I$.
  In other words, $\sigma\in A_I(r_l)$. Conversely $r_f(\sigma)=0$ implies
  $\sigma\notin A_I(r_l)$. Thus the set of active sequences is exactly the set
  of sequences that can be made optimal for some choice of utility function for
  the follower. Because $y$ is chosen to minimize the utility over active
  variables, this corresponds to the utility achieved when committing to $r_l$.

  Since every RSSE is a solution of the MIP, and the optimal solution to
  the MIP corresponds to the payoff received by the leader if they were to
  commit to the strategy computed by the MIP, we conclude that the MIP computes
  an RSSS. If not, there would exist some RSSE which achieves better utility
  than what is computed by the MIP. This would be a contradiction, since
  such an RSSE would also be feasible and its objective would be equal to the
  RSSE value.
\end{proof}
\section{MIP for Robust Limited-Lookahead Stackelberg Equilibria}
Here we present a MIP that computes an RSSS when the follower has limited
lookahead and there is interval uncertainty about the evaluation function used
to determine actions by the follower. This MIP is a straightforward combination
of the MIP for limited lookahead and the MIP for interval uncertainty that we
presented in the main paper.
\begin{align}
  & \label{eq:mip_obj}
  \qquad\min_{r,v,s,y} y_0 \span\span\\
  & y_{\infoset_f(\sigma_f)} \geq\sum_{\mathclap{\substack{I'\in \cI_f\\\seq_f(I')=\sigma_f}}} y_{I'} - \sum_{\sigma_l \in \Sigma_l} g_l(\sigma_l,\sigma_f) r_l(\sigma_l) - M(1-r_f(\sigma_f)) \span\nonumber\\ & \forall \sigma_f \in \Sigma_f &\label{eq:zerosum_dual_constraint}
  % & v_{\infoset_f(\sigma_f)} = s^-_{\sigma_f} + \sum_{\mathclap{\substack{I'\in \cI_f\\\seq_f(I')=\sigma_f}}} v^-_{I'} + \sum_{\sigma_l \in \Sigma} r_l(\sigma_l)g^-_f(\sigma_l,\sigma_f) \span\span 
\end{align}
%%%%%%%%%%%%%%%%%%%%%%%%%%%%%%%%%%%%%%%%%%%%%%%%%%%%%%%%%%%%%%%%%%%%%%%%%%%%%%%
% Infoset value upper/lower best responses
%%%%%%%%%%%%%%%%%%%%%%%%%%%%%%%%%%%%%%%%%%%%%%%%%%%%%%%%%%%%%%%%%%%%%%%%%%%%%%%
\begin{align}
  \label{eq:infoset_value_constraintll_robust}
  & v_{I,\infoset_f(\sigma_f)}^q = s^q_{I,\sigma_f} + \sum_{\mathclap{\substack{I'\in \cI_I\\\seq_f(I')=\sigma_f}}} v^q_{I,I'} + \sum_{\sigma_l \in \Sigma} r_l(\sigma_l)g^q_I(\sigma_l,\sigma_f), \span\nonumber\\
  && \forall I\in \mathcal{I}_f, \sigma_f \in \Sigma_I, q\in \{U,L\}  \\
  \label{eq:slack_constraints}
  &  0 \leq s^q_{I,\sigma_f} \leq M (1-b_I^q(\sigma_f))& \forall I\in\mathcal{I}_f, \sigma_f\in \Sigma_I, q\in \{U,L\} \\
  & \sum_{\mathclap{a\in A(I')}}b_I(\seq_f(I') a) = 1 & \forall I\in \mathcal{I}_f, q\in \{U,L\}, I'\in \cI_I \\ 
  & \ b_I^q(\sigma_f) \in \{0,1\} & \forall I\in\mathcal{I}_f, \sigma_f \in \Sigma_I, q\in \{U,L\}
\end{align}
\begin{align}
  &v_{I,I}^U - s_{I,\sigma}^U \geq v_{I,I}^L - M(1-r_f(\sigma)) \label{eq:optimal_actionsll_robust}\\
  & v_{I,I}^U - s_{I,\sigma}^U \leq v_{I,I}^L + Mr_f(\sigma)\label{eq:inoptimal_actionsll_robust}
\end{align}
\begin{align}
  \label{eq:empty_sequence_constraintll_robust}
  & r_i(\emptyset)=1, & \forall i\in \{l,f\} \\
  \label{eq:sequence_constraintsll_robust}
  & r_l(\sigma)=\sum_{a\in A(I)}r_l(\sigma a) & \forall I\in \cI_l, \sigma=\seq_l(I) \\ 
  & r_f(\sigma)<=\sum_{a\in A(I)}r_f(\sigma a) & \forall I\in \cI_f, \sigma=\seq_f(I)
\end{align}
\begin{align}
  & \ r_f(\sigma_f) \in \{0,1\} & \sigma_f \in \Sigma_f \\
  & 0 \leq r_l(\sigma_l) \leq 1 & \sigma_l \in \Sigma_l \label{eq:mip_probabilityll_robust}
\end{align}

\end{document}